\documentclass[journal]{IEEEtran}

\usepackage{graphicx}
\usepackage{amsmath,amsfonts,amssymb}
\usepackage{algorithm}
\usepackage{algpseudocode}
\usepackage{mathabx}
\usepackage{hyperref}
\usepackage{array}
\usepackage{textcomp}
\usepackage{stfloats}
\usepackage{url}
\usepackage{booktabs}
\usepackage{xcolor}
\usepackage{cite}
\usepackage{multirow}
\usepackage{bm}
\usepackage{comment}
\usepackage{pgfplots}
\pgfplotsset{compat=1.18}
\usepgfplotslibrary{groupplots}
\usepackage{amsthm}
\usepackage[siunitx,american]{circuitikz} 
\usetikzlibrary{arrows.meta, positioning, calc}

\definecolor{myblue}{RGB}{86, 180, 233} 
\newcommand{\tblue}[1]{\textcolor{myblue}{#1}}

\definecolor{mygreen}{RGB}{0, 158, 115} 
\newcommand{\tgreen}[1]{\textcolor{mygreen}{#1}}

\definecolor{myorange}{RGB}{230, 159, 0} 
\newcommand{\torange}[1]{\textcolor{myorange}{#1}}

\newtheorem{lemma}{Lemma}

\DeclareMathOperator{\sgn}{sgn}

\newcommand{\argmax}{\arg\!\max}

\title{Geometry-Informed Optimization of Binary RIS Configurations for Communication and Sensing}

\author{\IEEEauthorblockN{
Angelos Gkekas, 
Alexandros I. Papadopoulos,
Petros Andreas Pantazopoulos,\\
Antonios Lalas, 
Konstantinos Votis,
and Christos Liaskos 
}
\thanks{A. Gkekas, P. A. Pantazopoulos, A. Lalas and K. Votis are with the Information Technologies Institute, CERTH, Greece (e-mails: \{angelosgkekas, ppantazopoulos, lalas, kvotis\}@iti.gr).}
\thanks{A. Papadopoulos is with the Computer Science Engineering Department, University of Ioannina, Ioannina, Greece and with the Information Technologies Institute, CERTH, Greece (e-mail: a.papadopoulos@uoi.gr/alexpap@iti.gr).}
\thanks{C. K. Liaskos is with the Computer Science Engineering Department, University of Ioannina, Ioannina, and with the Foundation for Research and Technology Hellas (FORTH), Greece (e-mail: cliaskos@uoi.gr).}
\thanks{This work has received funding from the EU’s Horizon Europe research and innovation programme in the frame of the AutoTRUST project “Autonomous self-adaptive services for TRansformational personalized inclUsivenesS and resilience in mobility” under the Grant Agreement No 101148123.}
}

\begin{document}

\maketitle

\begin{abstract}

Practical Reconfigurable Intelligent Surfaces (RISs) often support only a small number of phase states, making their configuration inherently discrete. For a 1-bit RIS with \(N\) elements, direct optimization requires searching among \(2^N\) binary configurations. This work shows that this exponential configuration space is not unstructured. By reformulating 1-bit RIS optimization as the maximization of the norm of a signed sum of channel-dependent vectors, we prove that every globally optimal configuration must be induced by the signs of their projections onto a common direction. This geometric characterization restricts the class of configurations that can contain global optima and leads to different algorithmic consequences depending on the signal-space dimension. For general Multiple-Input-Multiple-Output (MIMO) systems, we develop a geometry-informed sampling method that evaluates only structurally admissible configurations. For Single-Input-Single-Output (SISO) systems, the same principle reduces to a two-dimensional angular partition, allowing the complete candidate set to be characterized and the global optimum to be recovered through polynomial-time enumeration, by evaluating at most \(N+1\) out of the \(2^N\) configurations. Finally, we apply the same binary optimization principle to an Integrated Sensing and Communication (ISAC) scenario, where communication enhancement and target localization reduce to the same underlying geometric problem. The proposed framework therefore provides a unified approach for exploiting the structure of practical 1-bit RIS configurations across communication and sensing functionalities.

\end{abstract}


\section{Introduction}\label{sec:intro}

The transition towards Beyond-5G (B5G) and 6G networks has introduced a fundamentally different perspective on the wireless propagation environment. Instead of merely adapting to the characteristics of the wireless channel, which have traditionally been regarded as inherent constraints of the communication system, emerging technologies enable direct interaction with the propagation medium itself~\cite{LiaskosEtAl}. Through the deployment of programmable electromagnetic structures, the propagation of wireless signals can be dynamically manipulated to achieve specific communication objectives. This capability has given rise to the concept of Programmable Wireless Environments (PWEs), in which software-controlled elements actively contribute to network optimization, creating new opportunities for enhancing coverage, reliability, spectral efficiency, and energy efficiency~\cite{LiaskosEtAl2}.

Reconfigurable Intelligent Surfaces (RISs) represent one of the key enabling technologies behind the realization of PWEs~\cite{PutrantoHamzaMabroukiArmiDayoup}. These structures consist of engineered planar arrays composed of a large number of low-cost, quasi-passive elements whose electromagnetic response can be electronically tuned~\cite{PapadopoulosEtAl}. By adjusting properties such as phase shift, amplitude, or polarization, RISs can control the way incident electromagnetic waves are reflected and redirected in space~\cite{PapadopoulosEtAl2}. Through this controllable interaction with the wireless medium, RISs transform the passive propagation environment into an active component of the communication system, enabling improvements in received signal quality, coverage, interference management, and energy efficiency~\cite{PapadopoulosEtAl3,WuZhang}. In terms of 6G services, RISs have also attracted significant interest in sensing applications and Integrated Sensing and Communication (ISAC) systems, where programmable propagation can enhance tasks such as localization, detection, and environmental awareness~\cite{BuzziGrossiLopsVenturino, TishchenkoEtAl}. 
The same principle can support RF sensing in other applications, such as automotive and public-transport environments. By creating propagation conditions that make people and their movements more observable, RISs could facilitate in-cabin tasks such as occupancy detection, people counting, and passenger characterization~\cite{MunteLazaroVillarino,SatoWandaleIchigeKimuraSugiura}.

Despite their flexibility, practical RIS implementations are subject to significant hardware constraints. In particular, the phase response of each RIS element is typically controlled through simple switching circuits that provide only a limited number of discrete phase states. While many theoretical studies assume continuously tunable phase shifts, commercial and experimentally demonstrated RIS platforms commonly operate with 1-bit or 2-bit phase quantization, corresponding to two or four available phase shift values~\cite{ShekhawatKashyapRaldirisShanTrichopoulos, ZhaoJianChenZhaoMu}. Binary RIS architectures, which restrict each element to two possible phase states (typically \(0^\degree\) and \(180^\degree\)), are especially attractive due to their low hardware complexity and ease of implementation~\cite{JohariEtAl}.

The discrete nature of practical RIS hardware gives rise to a challenging optimization problem. For a binary RIS consisting of \(N\) elements, the number of possible phase configurations grows exponentially as \(2^N\), causing the computational complexity of exhaustive search to scale prohibitively with the surface size. Consequently, determining the phase configuration that maximizes a desired performance metric, such as received signal strength, achievable rate, sensing accuracy, or detection performance, becomes a large-scale combinatorial optimization task.
Meanwhile, most existing optimization approaches have been developed under the assumption of continuous phase shifts. These methods typically formulate RIS configuration as a non-convex optimization problem and employ techniques such as semidefinite relaxation~\cite{WuZhang2,BuzziGrossiLopsVenturino}, alternating optimization~\cite{WuZhang2,HuangZapponeAlexandropoulosDebbahYuen}, manifold optimization~\cite{YuXuSchober}, or gradient-based algorithms~\cite{HuangZapponeAlexandropoulosDebbahYuen,PerovicTranRenzoFlanagan} to obtain approximate solutions. When applied to discrete RIS architectures, a common strategy is to first solve the continuous problem and subsequently quantize the resulting phases~\cite{HuangAlexandropoulosZapponeDebbahYuen,ZhaoWuZhaoZhang,YuanLiangJoungFengLarson}. However, the quantization step generally destroys optimality and may lead to substantial performance degradation when the phase resolution is low \cite{NiLiuYangTianShen}. Alternatively, one may directly search the discrete configuration space, but the exponential growth of the number of feasible solutions quickly renders unstructured search methods computationally impractical.

In this work, we pursue a different direction by addressing binary RIS phase optimization directly in the discrete domain. Rather than first solving a continuous-phase problem and subsequently quantizing its solution, we investigate the structure of the binary configuration space and identify a geometric property shared by its globally optimal solutions. We exploit this property to develop optimization methods adapted to different signal-space dimensions and show that the same underlying principle can support RIS phase design for both communication and sensing functionalities.

\subsection{Related Work}

RIS configuration optimization has been extensively studied for communication systems, with objectives including received signal strength maximization, achievable rate maximization, energy efficiency improvement, interference mitigation, and coverage enhancement. In \cite{WuZhang2}, the RIS phase optimization problem is formulated as a joint active--passive beamforming problem under continuous phase shifts. In \cite{PerovicTranRenzoFlanagan}, the achievable rate of an RIS-assisted MIMO system is maximized by jointly optimizing the transmit covariance matrix and the RIS phase shifts. The work \cite{HuangZapponeAlexandropoulosDebbahYuen} studies energy-efficient resource allocation by jointly optimizing the transmit power allocation and the continuous RIS phase shifts. The work \cite{YuXuSchober} considers spectral efficiency maximization through the joint optimization of the beamforming vector and the RIS phase shifts. In \cite{ZengZhangDiHanSong}, coverage enhancement is studied through the optimization of the RIS location and orientation.
A large part of this literature assumes continuously adjustable phase shifts and formulates RIS configuration as a non-convex optimization problem. Techniques such as semidefinite relaxation, alternating optimization, manifold optimization, and gradient-based methods are then employed to obtain approximate solutions.

The limited phase resolution of practical RIS hardware has motivated increasing interest in discrete optimization. A common strategy is to first solve the continuous-phase problem and subsequently quantize the resulting configuration to the available phase states~\cite{HuangAlexandropoulosZapponeDebbahYuen,ZhaoWuZhaoZhang,YuanLiangJoungFengLarson}. Although computationally convenient, such approaches optimize over a relaxed domain and impose the hardware constraints only after the optimization has been completed. Consequently, the resulting configuration is not, in general, optimal over the actual discrete feasible set.

A different line of work addresses the discrete problem directly. Existing approaches include discrete search procedures, greedy methods, branch-and-bound techniques, and other approximation strategies. For example, in \cite{WuZhang3}, an alternating optimization approach is proposed, where the RIS phase shifts are optimized one at a time while keeping the remaining phase shifts fixed. In \cite{YouZhengZhang}, a successive refinement algorithm is proposed to optimize discrete RIS phase shifts for maximizing the achievable rate and for minimizing the mean squared error of channel estimation. The work \cite{ZhangRenShenLuo} shows that the special case of the SISO system can be solved in polynomial time, without further developing this result into an explicit algorithm. In \cite{StylianopoulosGavriilidisAlexandropoulos}, a closed-form sign-alignment rule is proposed for the SISO 1-bit RIS phase optimization problem, yielding an asymptotically optimal approximation. These approaches explicitly account for the available phase states, but they typically rely on heuristic search or problem-specific approximations rather than utilizing structural properties of the solution space. While some works exploit properties of particular channel models, system dimensions, or optimization settings to reduce the computational complexity, the broader geometric structure of the configurations that can contain globally optimal solutions, particularly across different signal-space dimensions, has received considerably less attention.

Beyond communication-oriented objectives, RIS optimization has also been widely investigated for sensing and ISAC. Existing works have considered RIS-assisted target detection and radar sensing~\cite{BuzziGrossiLopsVenturino,VejlingKimBiscioWymeerschPopovski,ShaoYouMaChenZhang}, as well as localization and environment reconstruction~\cite{SotiropoulosEtAl,HuangYangTangWenXiaJin,ElzanatyGuerraGuidiAlouni}. Joint communication and sensing has also received significant attention, with studies addressing the trade-offs and optimization strategies arising when the same RIS infrastructure supports both functionalities simultaneously~\cite{TishchenkoEtAl,WuEtAl}. These studies generally formulate and solve the corresponding communication and sensing optimization problems separately. This motivates investigating whether different RIS functionalities can be described through a common underlying discrete optimization structure.

\subsection{Motivation and Contribution}
The combinatorial nature of binary RIS optimization makes exhaustive search impractical for realistically sized surfaces. Rather than asking how to explore the entire \(2^N\)-configuration space more efficiently, this work investigates whether all binary configurations are equally relevant candidates for global optimality. The resulting contributions are summarized as follows:

\begin{itemize}

\item We reformulate 1-bit RIS phase optimization for received-signal maximization as the maximization of the norm of a signed sum of channel-dependent real vectors. We prove that every globally optimal configuration is induced by the signs of the projections of these vectors onto a common direction. This result identifies a geometrically structured class of sign patterns that contains all global optima.

\item We propose a geometry-informed sampling algorithm for the general MIMO setting. The algorithm samples directions in the underlying signal space and converts them into binary RIS configurations through the derived projection rule. Consequently, every evaluated configuration satisfies the necessary geometric condition for global optimality, allowing the sampling budget to be concentrated on structurally admissible candidates.

\item We specialize the geometric characterization to the SISO setting, where the channel-dependent vectors lie in a two-dimensional real space. We construct the resulting angular partition, enumerate its distinct sign patterns, and recover a globally optimal binary configuration by evaluating \(K\) candidates, where \(K\leq N+1\) is the number of nonzero vectors after incorporating the direct channel.

\item We apply the proposed formulation to an RIS-assisted ISAC scenario in which the RIS elements are partitioned between communication and sensing. We show that the RIS phase-design subproblems for communication-signal enhancement and sensing illumination have the same signed-vector form and can therefore be addressed using the proposed MIMO or SISO optimization method. This establishes a common binary phase-design mechanism across the two functionalities.

\end{itemize}

\subsection{Structure}

The remainder of this paper is organized as follows. In Section~\ref{sec:system_model}, we present the system model underlying the proposed optimization framework. In Section~\ref{sec:optimization}, we derive structural properties of the optimal configurations and exploit them to develop efficient algorithms for the MIMO case (Section~\ref{sec:MIMO}) and the SISO case (Section~\ref{sec:siso}). Section~\ref{sec:isac} discusses the application of the proposed methods in an ISAC setting. Numerical results and performance evaluations are provided in Section~\ref{sec:results}, while Section~\ref{sec:conclusions} offers concluding remarks.

\section{System Model}\label{sec:system_model}
Consider an array of \(N_{\text{T}}\) transmitters and an array of \(N_\text{R}\) receivers located, respectively, at:
\begin{align}
    {\boldsymbol{\psi}}_{i}^\text{T} =& \left(x_{i}^\text{T}, y_{i}^\text{T}, z_{i}^\text{T}\right), \quad i \in \{1,2, \dots, N_{\text{T}}\}, \\
    {\boldsymbol{\psi}}_{j}^\text{R} =& \left(x_{j}^\text{R}, y_{j}^\text{R}, z_{j}^\text{R}\right), \quad j \in \{1,2, \dots, N_{\text{R}}\}.
\end{align}
Furthermore, consider an RIS composed of a rectangular \(N_1 \times N_2\) grid of cells, with distance \(d_\text{cell}\) between adjacent cells. The center of the RIS grid is located at:
\begin{equation}
{\boldsymbol{\psi}}^{\text{RIS}} = \left(x^{\text{RIS}}, y^{\text{RIS}}, z^{\text{RIS}}\right). 
\end{equation}
The orientation of the RIS is defined by two unit vectors \(\widehat{\mathbf{n}}\) and \(\widehat{\mathbf{t}}\), where \(\widehat{\mathbf{n}}\) is normal to the RIS surface and \(\widehat{\mathbf{t}}\) is tangent to the direction of the grid with size \(N_1\). The direction of the grid with size \(N_2\) is determined by the unit vector \(\widehat{\mathbf{n}} \times \widehat{\mathbf{t}}\). The number of cells in the RIS grid is \(N = N_1N_2\). We denote the center of each cell by:
\begin{equation}
{\boldsymbol{\psi}}_{n}^\text{RIS} = \left(x_{n}^\text{RIS}, y_{n}^\text{RIS},z_{n}^\text{RIS}\right), \quad n \in \{1,2,\dots, N\},
\end{equation}
where the index \(n\) follows an \(N_1\)-major ordering of the cells, and the first point \(\boldsymbol{\psi}_{1}^{\text{RIS}}\) corresponds to the center of the corner cell in the directions of \(\widehat{\mathbf{t}}\) and \(\widehat{\mathbf{n}} \times \widehat{\mathbf{t}}\) relative to \(\boldsymbol{\psi}^{\text{RIS}}\).

The phase shifts introduced by the RIS can take one of two values, \(0^\degree\) and \(180^\degree\). Therefore, each cell either leaves the corresponding signal component unchanged or multiplies it by \(-1\). Thus, the phase configuration of the RIS is described by a vector:
\begin{equation}
    \mathbf{x} \in \{+1, -1\}^N.
\end{equation}

Each transmitter emits a known signal with power \(P\) at carrier frequency \(f_c\), and the signals transmitted by different transmitters are orthogonal. We assume a normalized scalar free-space propagation model between all elements of the system, such that the channel between any two points \(\boldsymbol{\psi}_1\) and \(\boldsymbol{\psi}_2\) is described by:
\begin{equation}
\mathcal{C}(\boldsymbol{\psi}_1, \boldsymbol{\psi}_2) = \frac{e^{-j\frac{2\pi}{\lambda_c}\|\boldsymbol{\psi}_1 - \boldsymbol{\psi}_2\|}}{\|\boldsymbol{\psi}_1 - \boldsymbol{\psi}_2\|}, \label{eq:channel_model}
\end{equation}
where \(\lambda_c\) is the wavelength corresponding to \(f_c\).


Then, by considering the Tx \(\to\) RIS, RIS \(\to\) Rx, and direct Tx \(\to\) Rx propagation paths, the corresponding channel vectors between the transmitters, the receivers, and the \(n\)-th RIS cell are defined as:
\begin{align}
    \mathbf{h}_{n}^\text{T} &= \left(
    \mathcal{C}\!\left(\boldsymbol{\psi}_{1}^\text{T}, \boldsymbol{\psi}_{n}^\text{RIS}\right),
    \mathcal{C}\!\left(\boldsymbol{\psi}_{2}^\text{T}, \boldsymbol{\psi}_{n}^\text{RIS}\right),
    \dots,
    \mathcal{C}\!\left(\boldsymbol{\psi}_{N_\text{T}}^\text{T}, \boldsymbol{\psi}_{n}^\text{RIS}\right)
    \right)^\top, \label{eq:ch1}\\
    \mathbf{h}_{n}^\text{R} &= \left(
    \mathcal{C}\!\left(\boldsymbol{\psi}_{n}^\text{RIS}, \boldsymbol{\psi}_{1}^\text{R}\right),
    \mathcal{C}\!\left(\boldsymbol{\psi}_{n}^\text{RIS}, \boldsymbol{\psi}_{2}^\text{R}\right), 
    \dots, 
    \mathcal{C}\!\left(\boldsymbol{\psi}_{n}^\text{RIS}, \boldsymbol{\psi}_{N_\text{R}}^\text{R}\right)
    \right)^\top, \label{eq:ch2}\\
   \mathbf{h}^{\text{TR}}&= \left(\mathcal{C}\!\left(\boldsymbol{\psi}^\text{T}_1,\boldsymbol{\psi}^\text{R}_1\right),\mathcal{C}\!\left(\boldsymbol{\psi}^\text{T}_1,\boldsymbol{\psi}^\text{R}_2\right),\ldots, \mathcal{C}\!\left(\boldsymbol{\psi}^\text{T}_{N_\text{T}},\boldsymbol{\psi}^\text{R}_{N_\text{R}}\right)\right)^\top ,\label{eq:ch3}
\end{align}
with \(\mathbf{h}_{n}^\text{T} \in \mathbb{C}^{N_\text{T}}\), \(\mathbf{h}_{n}^\text{R} \in \mathbb{C}^{N_\text{R}}\), and \(\mathbf{h}^{\text{TR}} \in \mathbb{C}^{N_{\text{T}} N_{\text{R}}}\).

Finally, we collect the received Channel State Information (CSI) observations in the vector \(\mathbf{r} \in \mathbb{C}^{N_\text{T}N_\text{R}}\), given by:
\begin{equation}
\mathbf{r} =  \mathbf{h}^\text{TR} + \sum_{n=1}^N x_n \mathbf{h}^\text{T}_n \otimes \mathbf{h}^\text{R}_n + \mathbf{w}, \label{eq:r0}
\end{equation}
where \(\otimes\) denotes the Kronecker product, \(x_n\) is the \(n\)-th component of the phase configuration vector \(\mathbf{x}\), and \(\mathbf{w} \in \mathbb{C}^{N_\text{T}N_\text{R}}\) is Additive White Gaussian Noise (AWGN) with independent and identically distributed entries drawn from \(\mathcal{CN}(0, \sigma^2)\). The transmit Signal-to-Noise Ratio (SNR) is defined as:
\begin{equation}
\gamma = \frac{P}{\sigma^2}.
\end{equation}

The simplified propagation model adopted in Eq.~\eqref{eq:channel_model} provides a controlled and computationally tractable setting in which the proposed binary optimization framework can be isolated and evaluated. The model captures the dominant phase and distance-dependent attenuation effects of the considered propagation paths, but does not account for all phenomena that may arise in complex indoor environments, such as diffuse scattering, multi-bounce reflections, rapidly varying blockages, mutual coupling among RIS elements, or polarization-dependent responses. Therefore, the simulations presented in this work are intended to evaluate the proposed optimization principle under a structured channel model, rather than to provide a full-wave characterization of a specific deployment environment.

\section{Proposed Optimization Framework}\label{sec:optimization}

The RIS can modify the propagation channel and support multiple functionalities, but determining the corresponding phase configuration typically leads to challenging optimization problems. In this work, we focus on a class of communication and sensing functionalities that can be formulated through LOS signal-strength maximization. Our objective is not merely to solve the resulting binary optimization problem, but to reveal the geometric structure satisfied by its optimal solutions and exploit this structure algorithmically. Initially, we derive a general mathematical formulation of the underlying problem. We then proceed with the development of optimization methods for the MIMO and SISO settings.

\subsection{Optimization Problem Formulation}\label{}

Consider the problem of maximizing the SNR at the receiver over the phase configuration of the RIS. Since the AWGN term is independent of the RIS configuration, maximizing the receiver SNR is equivalent to maximizing the norm of the noise-free CSI vector \(\mathbf{r}\), as described by Eq.~\eqref{eq:r0}. The following sequence of simplifications yields a more compact formulation of the problem, which reveals its underlying mathematical structure:

\begin{enumerate}
    \item Express each complex number \(z \in \mathbb{C}\) as a pair \((x,y) \in \mathbb{R}^2\), and interpret each vector of \(\mathbb{C}^{N_\text{T}N_\text{R}}\) as an element of \(\mathbb{R}^{2N_\text{T}N_\text{R}}\). The problem becomes:
    \begin{equation}
        \mathbf{x}^* \in \underset{\mathbf{x} \in \{+1,-1\}^N}{\mathbf{\argmax}} \left\|\mathbf{c} + \sum_{n=1}^N x_n \mathbf{a}_n \right\|,
    \end{equation}
    where \(\mathbf{c}\) and \(\mathbf{a}_n\) are elements of \(\mathbb{R}^{2N_\text{T}N_\text{R}}\), and correspond to 
    \(\mathbf{h}^{\text{TR}}\) and \(\mathbf{h}^\text{T}_n \otimes \mathbf{h}^\text{R}_n\), respectively.

    \item Integrate the constant vector \(\mathbf{c}\) inside the sum, by introducing a fresh variable \(x_{N+1} \in \{+1,-1\}\) and setting \(\mathbf{a}_{N+1} = \mathbf{c}\). This results in a homogeneous form of the problem:
    \begin{equation}\label{eq:problem0}
    \mathbf{x}^* \in \underset{\mathbf{x} \in \{+1,-1\}^{N+1}}{\mathbf{\argmax}} \left\|\sum_{n} x_n \mathbf{a}_n \right\| ,
    \end{equation}
    where an additional constraint \(x_{N+1} = 1\) is needed to ensure that the two objective functions are equal. Notice that the constraint \(x_{N+1} = 1\) does not add any extra complexity. For any solution \(\mathbf{x}^*\) of the unconstrained problem described in Eq.~\eqref{eq:problem0}, \(-\mathbf{x}^*\) is also a solution, as the negative sign is absorbed by the norm. Therefore, we can solve the unconstrained problem, and simply change the signs of all variables \(x_n\) afterwards, should the constraint \(x_{N+1} = 1\) be violated.

    \item Assume that all vectors \(\mathbf{a}_n\) are non zero, since any \(\mathbf{a}_n = \mathbf{0}\) does not contribute to the sum and can be disregarded. This step reduces the size of the problem, and also achieves a technical requirement needed for the analysis. By ignoring the zero vectors, we arrive at the final form of the optimization problem:
    \begin{equation}\label{eq:problem}
    \mathbf{x}^* \in \underset{\mathbf{x} \in \{+1,-1\}^{K}}{\mathbf{\argmax}} \left\|\sum_{k} x_k \mathbf{a}_k \right\|,
    \end{equation}
    where \(K \leq N+1\) is the number of nonzero vectors in the set \(\{\mathbf{a}_1, \mathbf{a}_2, \dots, \mathbf{a}_N, \mathbf{c}\}\) and \(\mathbf{0} \neq \mathbf{a}_k \in \mathbb{R}^d\), \(d = 2N_\text{T}N_\text{R}\). Here, for notational convenience, we have implicitly re-indexed the vectors so that the non-zero \(\mathbf{a}_k\) are exactly the first \(K\) vectors in the sequence \(\mathbf{a}_1, \mathbf{a}_2, \dots , \mathbf{a}_{N+1}\).
\end{enumerate}

\subsection{Informed Sampling Algorithm}\label{sec:MIMO}

The problem described by Eq.~\eqref{eq:problem} is an NP-hard combinatorial optimization problem, implying that exact tractable solutions are unlikely to exist in the general case, while related quadratic programming formulations and polynomially solvable special cases are presented in \cite{LiSunGuGaoLiu}. For the general problem considered here, a simple approximation approach could sample \(N_{\text{sample}}\) configurations from \(\{+1,-1\}^K\) and return the one achieving the highest norm. However, such an approach treats the binary configuration space as unstructured. The key observation of this work is that every globally optimal solution to Eq.~\eqref{eq:problem} must satisfy a specific geometric condition. This condition identifies a restricted class of configurations that contains all potentially optimal solutions. We exploit this structure through geometry-informed sampling in the general MIMO case, while in the SISO case the corresponding candidate set can be characterized explicitly, leading to an exact polynomial-time solution.

In this section, we exploit the geometric structure of the solutions to Eq.~\eqref{eq:problem}. Rather than sampling arbitrary configurations from \(\{+1,-1\}^K\), the proposed method restricts the search to a subclass of configurations that can contain globally optimal solutions.

At an intuitive level, this structure corresponds to aligning all vectors \(\mathbf{a}_k\) with a common direction in \(\mathbb{R}^d\). More precisely, if a configuration \(\mathbf{x}^*\) is a solution to the problem, then there must exist a vector \(\mathbf{v} \in \mathbb{R}^d\) such that, for all \(k\), \(x_k \mathbf{a}_k\) lies in the half space defined by \(\mathbf{v}\). This means that if \(\mathbf{a}_k\) and \(\mathbf{v}\) form an acute angle, then \(x_k = +1\), and if the angle is obtuse, then \(x_k = -1\), flipping \(\mathbf{a}_k\) and ``aligning" it to \(\mathbf{v}\). This is formally expressed by Lemma~\ref{lem:essential_x}.

\begin{lemma}\label{lem:essential_x}
    Let \(\mathbf{x}^* \in {\mathbf{\argmax}}_{\mathbf{x} \in \{+1,-1\}^{K}} \left\| \sum_k x_k \mathbf{a}_k\right\|\) with \(\mathbf{a}_k \neq \mathbf{0}\) for all \(k\). Then there exists some \(\mathbf{v} \in \mathbb{R}^d\) with \(\mathbf{a}_k \cdot \mathbf{v} \neq 0\) and \(x_k^* = \sgn(\mathbf{a}_k \cdot \mathbf{v})\) for all \(k\).
\end{lemma}

\begin{proof}
    Let \(\mathbf{x}^*\in {\mathbf{\argmax}}_{\mathbf{x} \in \{+1,-1\}^{K}} \left\| \sum_k x_k \cdot \mathbf{a}_k\right\|\) and set \(\mathbf{v}^* =  \sum_{k} x_k^* \mathbf{a}_k\) as the corresponding vector with maximum norm. We show that \(\mathbf{v}^*\) satisfies the requirements of the Lemma. Suppose, for the sake of contradiction, that for some \(i\), either \(\mathbf{a}_i \cdot \mathbf{v}^* = 0\) or \(x_i^* \neq \sgn( \mathbf{a}_i \cdot \mathbf{v}^* )\). Define \(\mathbf{x}'\) by \(x_i' = - x_i^*\) and \(x_k' = x_k^*\) for \(k \neq i\). Then:
    \begin{equation}
        \mathbf{v}' =  \sum_{k=1}^K x_k' \mathbf{a}_k = \mathbf{v}^* - 2 x_i^*  \mathbf{a}_i .
    \end{equation}
    The norm of \(\mathbf{v}'\) is described by:
    \begin{align}
        \|\mathbf{v}'\|^2  &= \|\mathbf{v}^*\|^2 + 4\|\mathbf{a}_i\|^2 -4x_i^* \mathbf{a}_i \cdot \mathbf{v}^* \label{eq:ineq1}\\
        &> \|\mathbf{v}^*\|^2 -4x_i^* \mathbf{a}_i \cdot \mathbf{v}^* \label{eq:ineq2} .
    \end{align}
    
    If \(\mathbf{a}_i \cdot \mathbf{v}^* = 0\) then we immediately get \(\|\mathbf{v}'\| > \|\mathbf{v}^*\|\), contradicting the maximality of \(\|\mathbf{v}^*\|\). On the other hand, if \(x_i^* \neq \sgn(\mathbf{a}_i \cdot \mathbf{v}^*)\), then \(x_i^* \mathbf{a}_i \cdot \mathbf{v}^* < 0\) and from \eqref{eq:ineq2} we again get \(\|\mathbf{v}'\|> \|\mathbf{v}^*\|\). Both cases lead to a contradiction, meaning that we must indeed have \(\mathbf{a}_k \cdot \mathbf{v}^* \neq 0\) and \(x_k^* = \sgn(\mathbf{a}_k \cdot \mathbf{v}^*)\), for all \(k\).
    
\end{proof}

Lemma~\ref{lem:essential_x} provides a direct mechanism for generating structurally admissible configurations. We first sample a vector \(\mathbf{v} \in \mathbb{R}^d\) from a circularly symmetric distribution and then construct the corresponding configuration as \(x_k = \sgn(\mathbf{a}_k \cdot \mathbf{v})\), \(k \in \{1, 2, \dots , K\}\). In this way, every sampled configuration satisfies the geometric condition of Lemma~\ref{lem:essential_x}, while configurations that cannot be globally optimal are excluded from the search. The procedure is presented in Algorithm~\ref{alg:better}.

\begin{algorithm}[ht]
    \caption{Informed Sampling Approach}\label{alg:better}
    \hspace*{\algorithmicindent} \textbf{Input:} \(\mathbf{a}_1, \mathbf{a}_2, \dots, \mathbf{a}_K \in \mathbb{R}^d - \{\mathbf{0}\}, \ N_{\text{samples}} \in \mathbb{N} - \{0\}\) \\
    \hspace*{\algorithmicindent} \textbf{Output:} \(\mathbf{x} \in \{+1,-1\}^K \text{ maximizing } \left\| \sum_k x_k \mathbf{a}_k\right\| \) 
    \begin{algorithmic}[1]
    
    \State \(d_{\max} \gets 0\)
    \For{\(i \in \{1,2, \dots, N_{\text{samples}}\}\)}
        \State \textbf{sample} \(\mathbf{v}\) \textbf{from} \(\mathbb{R}^d\)
        \State \(\mathbf{x} \gets (\sgn(\mathbf{a}_1 \cdot \mathbf{v}), \sgn(\mathbf{a}_2 \cdot \mathbf{v}), \dots, \sgn(\mathbf{a}_k \cdot \mathbf{v}))\)
        \If{\(\left\| \sum_k x_k  \mathbf{a}_k \right\| > d_{\max}\)}
            \State \(d_{\max} \gets \left\| \sum_k x_k  \mathbf{a}_k \right\|\)
            \State \(\mathbf{x}^* \gets \mathbf{x}\)
        \EndIf
    \EndFor
    \State \Return \(\mathbf{x}^*\)
    
    \end{algorithmic}
\end{algorithm}

By restricting the search to configurations satisfying Lemma~\ref{lem:essential_x}, the proposed method uses the available sampling budget exclusively on structurally admissible candidates, in contrast to direct sampling over the entire \(\{+1,-1\}^K\) space.

\subsection{Exact Polynomial-Time Solution for SISO Systems}\label{sec:siso}

The geometric characterization of Lemma~\ref{lem:essential_x} also reveals how the structure of the problem changes with the signal-space dimension. In the general MIMO case, the admissible configurations are induced by directions in a high-dimensional space and are therefore explored through sampling. In the SISO case, however, the vectors \(\mathbf{a}_k\) lie in \(\mathbb{R}^2\), allowing the complete set of geometrically admissible configurations to be characterized explicitly. As a result, the randomized strategy of Algorithm~\ref{alg:better} can be replaced by an exact polynomial-time algorithm. Thus, the SISO solution is not based on a separate optimization principle, but emerges as the low-dimensional specialization of the same geometric structure used for the general MIMO problem.

More specifically, Lemma~\ref{lem:essential_x} implies that every optimal configuration \(\mathbf{x}^* \in \{+1, -1\}^K\) obeys the rule \(x_k = \sgn(\mathbf{a}_k \cdot \mathbf{v})\) for some \(\mathbf{v} \in \mathbb{R}^2\). The sign of any of these dot products only changes when \(\mathbf{v}\) crosses the line orthogonal to \(\mathbf{a}_k\), meaning that there is a one-to-one correspondence between the induced configurations \(\mathbf{x}\) and the sectors of \(\mathbb{R}^2\) defined by the lines \(l_k\) = \(\{\mathbf{p} | \mathbf{a}_k \cdot \mathbf{p} = 0\}\), \(k \in \{1,2,\ldots,K\}\). It is easy to show that the number of sectors, and consequently the number of configurations to be examined, is at most \(2K\) (exactly \(2K\) when all vectors \(\mathbf{a}_k\) are pairwise linearly independent, and less than \(2K\) otherwise). 
This can be evidenced in Fig.~\ref{fig:sectors} in the simple case of \(K=3\), with the generalization to any \(K\) being straightforward. 
Finally, we can determine the \(2K\) possible \(\mathbf{x}\) configurations by inspecting whether a vector \(\mathbf{v}\) lying inside each sector belongs to the half-plane defined by the vector \(\mathbf{a}_k\) and the line \(l_k\). If it belongs then \(x_k=+1\), otherwise \(x_k=-1\). This simple and geometry-assisted argument is visualized in Fig.~\ref{fig:sectors}, which shows the possible \(\mathbf{x}\) configurations. 

What is more, we can further reduce the number of configurations to be evaluated. It can be observed in Fig.~\ref{fig:sectors} that configurations corresponding to antipodal sectors are negatives of one another. However, these configurations are equivalent since the opposite sign is absorbed by the norm, and thus they yield the same objective value. Therefore, we only need to examine at most \(K\) configurations.

\begin{figure}[ht]
\centering
\begin{tikzpicture}[scale=1.9]


\def \thetaa{35}
\def \thetab{110}
\def \thetac{150 + 180}

\def \r{1.8}

\def \L{.5}

\draw[myblue, thick] ({-\r*sin(\thetaa)}, {\r*cos(\thetaa)}) -- ({\r*sin(\thetaa)}, {-\r*cos(\thetaa)}) node[pos=1, anchor=west, text=black] {\(l_1\)};
\draw[myorange, thick] ({-\r*sin(\thetab)}, {\r*cos(\thetab)}) -- ({\r*sin(\thetab)}, {-\r*cos(\thetab)}) node[pos=1, anchor=west, text=black] {\(l_2\)};
\draw[mygreen, thick] ({-\r*sin(\thetac)}, {\r*cos(\thetac)}) -- ({\r*sin(\thetac)}, {-\r*cos(\thetac)}) node[pos=1, anchor=west, text=black] {\(l_3\)};

\draw[->, myblue, thick] ({-sin(\thetaa)*(0.8*\r+0.2)}, {cos(\thetaa)*(0.8*\r+0.2)}) -- ({-sin(\thetaa)*(0.8*\r+0.2) - \L*cos(\thetaa)},{cos(\thetaa)*(0.8*\r+0.2) - \L*sin(\thetaa)}) node[pos=1, anchor=east, text=black] {\(\mathbf{a}_1\)};
\draw[->, myorange, thick] ({-sin(\thetab)*(0.8*\r+0.2)}, {cos(\thetab)*(0.8*\r+0.2)}) -- ({-sin(\thetab)*(0.8*\r+0.2) + \L*cos(\thetab)},{cos(\thetab)*(0.8*\r+0.2) + \L*sin(\thetab)}) node[pos=1, anchor=west, text=black] {\(\mathbf{a}_2\)};
\draw[->, mygreen, thick] ({-sin(\thetac)*(0.8*\r+0.2)}, {cos(\thetac)*(0.8*\r+0.2)}) -- ({-sin(\thetac)*(0.8*\r+0.2) + \L*cos(\thetac)},{cos(\thetac)*(0.8*\r+0.2) + \L*sin(\thetac)}) node[pos=1, anchor=west, text=black] {\(\mathbf{a}_3\)};

\node at ({(0.5*\r + 0.55)*cos((\thetaa+\thetab)/2+90)}, {(0.5*\r + 0.5)*sin((\thetaa+\thetab)/2+90)}) {\scriptsize{\(\mathbf{x} = (\mathbf{\tblue{+1},\torange{+1},\tgreen{-1}})\)}};
\node at ({(0.5*\r + 0.55)*cos((\thetaa+\thetab)/2-90)}, {(0.5*\r + 0.5)*sin((\thetaa+\thetab)/2-90)}) {\scriptsize{\(\mathbf{x} = (\mathbf{\tblue{-1},\torange{-1},\tgreen{+1}})\)}};
    
\node at ({(0.5*\r + 0.45)*cos((\thetab+\thetac+180)/2+90) + 0.15}, {(0.5*\r + 0.45)*sin((\thetab+\thetac+180)/2+90)}) {\footnotesize{\(\mathbf{x} = (\mathbf{\tblue{-1},\torange{+1},\tgreen{+1}})\)}};
\node at ({(0.5*\r + 0.45)*cos((\thetab+\thetac+180)/2-90) - 0.2}, {(0.5*\r + 0.45)*sin((\thetab+\thetac+180)/2-90)}) {\footnotesize{\(\mathbf{x} = (\mathbf{\tblue{+1},\torange{-1},\tgreen{-1}})\)}};
    
\node at ({(0.5*\r + 0.25)*cos((\thetaa+\thetac)/2+90) +0.06}, {(0.5*\r + 0.25)*sin((\thetaa+\thetac)/2+90)}) {\scriptsize{\(\mathbf{x} = (\mathbf{\tblue{+1},\torange{-1},\tgreen{+1}})\)}};
\node at ({(0.5*\r + 0.25)*cos((\thetaa+\thetac)/2-90)}, {(0.5*\r + 0.25)*sin((\thetaa+\thetac)/2-90)}) {\scriptsize{\(\mathbf{x} = (\mathbf{ \tblue{-1},\torange{+1},\tgreen{-1} })\)}};

\fill (0,0) circle (0.03);

\draw[dash pattern=on 2pt off 3pt] (-2.2,0) -- (2.2,0);
\draw[dash pattern=on 2pt off 3pt] (0,-1.85) -- (0,1.85);

\end{tikzpicture}

\caption{Sectors and corresponding \(\mathbf{x}\) configurations for \(K = 3\), where each component \(x_k\) obeys the rule \(x_k = \sgn(\mathbf{a}_k \cdot \mathbf{v})\), with \(\mathbf{v}\) a vector in the appropriate sector. Each line \(l_k\), \(k = 1,2,3\), is orthogonal to the vector \(\mathbf{a}_k\).}
\label{fig:sectors}

\end{figure}
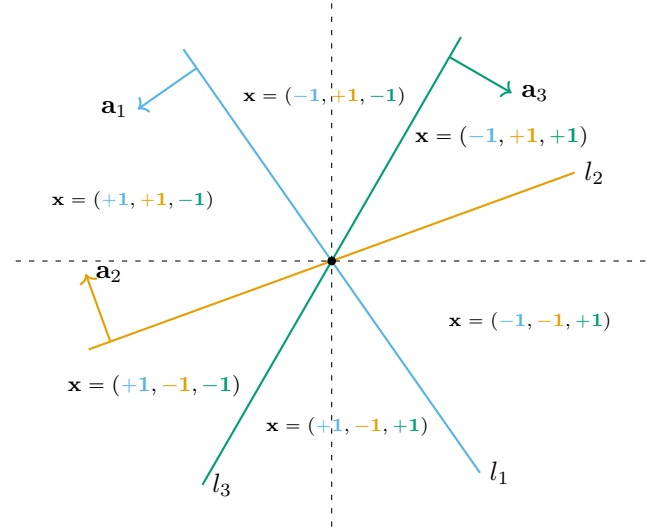

The complete optimization procedure is summarized in Algorithm~\ref{alg:plane}. We begin by constructing a list \(\Theta\), which contains the polar angles defining the boundaries of all sectors. Initially, \(\Theta\) contains exactly \(2K\) elements, corresponding to two angles for each vector \(\mathbf{a}_k\). Specifically, for each vector, the associated angles are those orthogonal to its direction and are given by \(\arg(\mathbf{a}_k)+\pi/2\) and \(\arg(\mathbf{a}_k)-\pi/2\), where \(\arg(\cdot)\) denotes the polar angle of a vector in \(\mathbb{R}^2\).

If all vectors \(\mathbf{a}_k\) have distinct orientations, then all \(2K\) elements in \(\Theta\) are unique. However, when there exist parallel vectors, some of the corresponding angles coincide, resulting in duplicate entries. Therefore, we remove all duplicate angles from \(\Theta\), yielding the final list of unique angles that determine the sector boundaries. We then perform a sorting of the elements of \(\Theta\), which results in each pair of consecutive entries \((\Theta_{j}, \Theta_{j+1})\) defining a sector boundary. As discussed previously, it is sufficient to examine only half of the resulting sectors. Therefore, we iterate over the indices corresponding to the first half of \(\Theta\), and for each index \(j\), we construct the \(\mathbf{x}\) configuration associated with the sector bounded by \((\Theta_j, \Theta_{j+1})\). A representative point within this sector is selected as the unit vector with polar angle \(\frac{\Theta_j+\Theta_{j+1}}{2}\):
\begin{equation}
\mathbf{v}_j = \left(\cos\left(\frac{\Theta_j + \Theta_{j+1}}{2}\right), \sin\left(\frac{\Theta_j + \Theta_{j+1}}{2}\right)\right).
\end{equation}
Then, the corresponding \(\mathbf{x}\) configuration is given by:
\begin{equation}
\mathbf{x}_j = (\sgn(\mathbf{a}_1 \cdot \mathbf{v}_j), \sgn(\mathbf{a}_2 \cdot \mathbf{v}_j), \dots, \sgn(\mathbf{a}_K \cdot \mathbf{v}_j)).
\end{equation}
The algorithm evaluates all \(\mathbf{x}\) configurations generated from the first half of the sectors, and returns the one that maximizes the objective value.

\begin{algorithm}[ht]
    \caption{Fast Exact Algorithm for \(d=2\)}\label{alg:plane}
    \hspace*{\algorithmicindent} \textbf{Input:} \(\mathbf{a}_1, \mathbf{a}_2, \dots, \mathbf{a}_K \in \mathbb{R}^2 - \{\mathbf{0}\}\) \\
    \hspace*{\algorithmicindent} \textbf{Output:} \(\mathbf{x} \in \{+1,-1\}^K \text{ maximizing } \left\|\sum_k x_k \mathbf{a}_k\right\| \) 
    \begin{algorithmic}[1]
    
    \State \(\Theta \gets \langle \rangle\) \Comment{Initialize empty list of angles}
    \For{\(k \in \{1, 2, \dots, K\}\)}
    \State \text{append} \(\arg(\mathbf{a}_k) + \frac{\pi}{2}  \pmod{\pi}\) \text{to} \(\Theta\)
    \State \text{append} \(\arg(\mathbf{a}_k) - \frac{\pi}{2}  \pmod{\pi}\) \text{to} \(\Theta\)
    \EndFor
    \State \text{sort} \(\Theta\) \Comment{\(\Theta\) has exactly \(2K\) elements}
    \State \text{remove duplicates from} \(\Theta\)  \Comment{\(\Theta\) has at most \(2K\) elements}
    \State \(d_{\max} \gets 0\)
    \For{\(j \in \{1,2, \dots, \textsc{Length}(\Theta)/2\}\)}
        \State \(\theta \gets \frac{\Theta_j + \Theta_{j+1}}{2}\)
        \State \(v \gets (\cos\theta, \sin\theta)\)
        \State \(\mathbf{x} \gets (\sgn(\mathbf{a}_1 \cdot \mathbf{v}), \sgn(\mathbf{a}_2 \cdot \mathbf{v}), \dots, \sgn(\mathbf{a}_k \cdot \mathbf{v}))\)
        \If{\(\left\| \sum_k x_k \mathbf{a}_k  \right\| > d_{\max}\)}
            \State \(d_{\max} \gets \left\| \sum_k x_k \mathbf{a}_k  \right\|\)
            \State \(\mathbf{x}^* \gets \mathbf{x}\)
        \EndIf
    \EndFor
    \State \Return \(\mathbf{x}^*\)
    
    \end{algorithmic}
\end{algorithm}

The computationally most expensive part of the algorithm is calculating all norms in lines \(13\) and \(14\), since we would need a loop repeating \(K\) times for calculating a single norm, and a total of \(\textsc{Length}(\Theta)/2 = \mathcal{O}(K)\) such calculations must be performed. This brings the time complexity of the algorithm to \(\mathcal{O}(K^2)\). If we assume that each norm, and each sign vector \(\mathbf{x}\) in line 12, can be computed in constant time by employing parallelism, then the loop in lines \(9-17\) can be executed in linear time, making the sorting of the list \(\Theta\) the dominant computational cost. Hence the time complexity reduces to \(\mathcal{O}(K\log K)\).

\section{Application to ISAC}\label{sec:isac}

Section~\ref{sec:optimization} focused on optimizing the 1-bit RIS phases in order to facilitate a pure communication task. In this section, we extend the framework to an ISAC scenario, where the RIS simultaneously supports both communication and sensing functionalities. This is achieved by partitioning the RIS elements into two sets: one dedicated to communication and one dedicated to sensing. Importantly, the RIS phase optimization problem arising in the sensing setup has the same mathematical form as Eq.~\eqref{eq:problem}, allowing the algorithms developed in Section~\ref{sec:optimization} to be directly applied in this setting as well.

We extend the system model with some additional components corresponding to the sensing functionality. We define a sensing region represented by the three-dimensional box:
\begin{equation}
    A = [x_{\min}, x_{\max}] \times [y_{\min}, y_{\max}] \times [z_{\min}, z_{\max}] .
\end{equation}
A single point target is assumed to lie within the sensing region, and the objective of the sensing procedure is to estimate its position. The true target location is denoted by:
\begin{equation}
    \boldsymbol{\psi}_{t}^{\text{real}} = \left(x_{t}^{\text{real}}, y_{t}^{\text{real}}, z_{t}^{\text{real}}\right) .
\end{equation}
Similarly to before, we suppose that the only signal propagation paths are Tx \(\to\) RIS, RIS \(\to\) Rx, Tx \(\to\) Rx, Tx \(\to\) target, RIS \(\to\) target, and target \(\to\) Rx, as shown in Fig.~\ref{fig:paths}. To this end, in addition to the channel vectors of Eqs.~\eqref{eq:ch1}-\eqref{eq:ch3} we consider:
\begin{align}
    \mathbf{h}^\text{T}_{t,\text{real}}  &= \left(
    \mathcal{C}\!\left(\boldsymbol{\psi}^\text{T}_1,\boldsymbol{\psi}_{t}^{\text{real}}\right),
    \mathcal{C}\!\left(\boldsymbol{\psi}^\text{T}_2,\boldsymbol{\psi}_{t}^{\text{real}}\right),\dots,
    \mathcal{C}\!\left(\boldsymbol{\psi}^\text{T}_{N_\text{T}},\boldsymbol{\psi}_{t}^{\text{real}}\right)
    \right)^\top  , \\
    \mathbf{h}^\text{R}_{t,\text{real}}  &= \left(
    \mathcal{C}\!\left(\boldsymbol{\psi}_{t}^{\text{real}},\boldsymbol{\psi}^\text{R}_1\right),\mathcal{C}\!\left(\boldsymbol{\psi}_{t}^{\text{real}},\boldsymbol{\psi}^\text{R}_2\right),
    \dots,
    \mathcal{C}\!\left(\boldsymbol{\psi}_{t}^{\text{real}},\boldsymbol{\psi}^\text{R}_{N_\text{R}}\right)
    \right)^\top ,
\end{align}
with \(\mathbf{h}^\text{T}_{t,\text{real}} \in \mathbb{C}^{N_\text{T}}\) and \(\mathbf{h}^\text{R}_{t,\text{real}} \in \mathbb{C}^{N_\text{R}}\).

\begin{figure}[ht]
\centering

\begin{circuitikz}[scale = 0.75]
\tikzstyle{every node}=[font=\LARGE]
\draw (2.5,7) to (2.5,7.5) node[dinantenna, rotate=-360]{};
\draw (12.5,7) to (12.5,7.5) node[dinantenna, rotate=-360]{};
\draw [line width=0.7pt, short] (3.75,12) -- (5,12);
\draw [ line width=0.7pt ] (7.5,9) circle (0.25cm);
\draw [line width=0.7pt, ->, >=Stealth, dashed] (2.6,8.8) -- (4.2,11.75); 
\draw [line width=0.7pt, ->, >=Stealth, dashed] (4.4,11.75) -- (7.25,9.25); 
\draw [line width=0.7pt, ->, >=Stealth, dashed] (3,8.35) -- (7.1,9); 
\draw [line width=0.7pt, ->, >=Stealth, dashed] (8,8.8) -- (12,8.3); 
\draw [line width=0.7pt, ->, >=Stealth, dashed] (4.6, 11.75) -- (12,8.6); 
\draw [line width=0.7pt, ->, >=Stealth, dashed] (3,8) -- (12,8); 
\node [font=\large] at (2.5,6.5) {Tx};
\node [font=\large] at (12.5,6.5) {Rx};
\node [font=\large] at (8.5,9.25) {target};
\node [font=\large] at (4.375,12.5) {RIS};
\node [font=\normalsize] at (2.75,10) {};
\node [font=\normalsize] at (6.25,10.5) {};
\node [font=\normalsize] at (4.75, 8) {};
\node [font=\normalsize] at (10.25,8) {};
\end{circuitikz}
\caption{Propagation paths considered in the channel model for the RIS-assisted ISAC scenario. The dashed arrows trace the Tx\(\to\)Rx, Tx\(\to\)RIS\(\to\)Rx, Tx\(\to\)target\(\to\)Rx, and Tx\(\to\)RIS\(\to\)target\(\to\)Rx links, accounting for the direct, RIS-assisted, and target-reflected components.}
\label{fig:paths}
\end{figure}
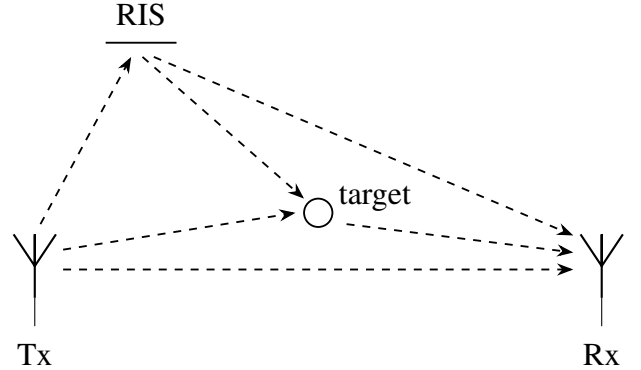

Then, the CSI vector can be expressed as the superposition of two components, one corresponding to paths that involve reflections from the target (Tx \(\to\) target \(\to\) Rx and Tx \(\to\) RIS \(\to\) target \(\to\) Rx), and another corresponding to paths that do not involve the target (Tx \(\to\) Rx and Tx \(\to\) RIS \(\to\) Rx). We refer to the first one as \(\mathbf{r}_{\text{sen}}\) and to the second as \(\mathbf{r}_\text{com}\). Therefore, the CSI vector is:\begin{align}
&\mathbf{r}_\text{sen} = \alpha\Big( \mathbf{h}^\text{T}_{t,\text{real}} \otimes \mathbf{h}^\text{R}_{t,\text{real}} \!+\!\! \sum_{n=1}^N\! x_n \mathcal{C}\!\left(\boldsymbol{\psi}^\text{RIS}_n, \boldsymbol{\psi}_{t,\text{real}}\right) \mathbf{h}^\text{T}_{n}
\otimes \mathbf{h}^\text{R}_{t,\text{real}} \Big)
\nonumber , \\
&\mathbf{r}_\text{com} = \mathbf{h}^\text{TR} + \sum_{n=1}^N x_n \mathbf{h}^\text{T}_n \otimes \mathbf{h}^\text{R}_n \label{eq:r} , \\
&\mathbf{r}_\text{total} = \mathbf{r}_{\text{sen}} + \mathbf{r}_\text{com} + \mathbf{w}\nonumber ,
\end{align}
where \(\alpha\) is an unknown scalar parameter that models the reflection coefficient of the target.

We partition the RIS cell indices \(\{1,2, \dots, N\}\) into two sets, \(S_1\) and \(S_2\). The cells indexed by \(S_1\) are allocated for optimizing communication performance, while those indexed by \(S_2\) are allocated for optimizing sensing. In the following sections, we describe the communication and the sensing scenarios.

\subsection{Communication Setup}\label{sec:comm_setup}

Since \(S_1\) corresponds to the cells that are assigned to optimizing the communication, we define the effective LOS signal by taking into account the true LOS and the reflections from the cells that correspond to \(S_1\):
\begin{equation}
\mathbf{r}_{\text{LOS}} = \mathbf{h}^\text{TR} + \sum_{n \in S_1} x_n \mathbf{h}^\text{T}_n \otimes \mathbf{h}^\text{R}_n .
\end{equation}
The communication optimization task is then identified as maximizing over \(x_n, n \in S_1\) the norm of \(\mathbf{r}_{\text{LOS}}\):
\begin{equation}
\mathbf{x}^*_\text{com} \in \underset{\substack{x_n \in \{+1,-1\} \\ n \in S_1}}{\mathbf{\argmax}} \|\mathbf{r}_\text{LOS}\| .
\end{equation}

This optimization problem has the same form as the problem defined in Eq.~\eqref{eq:problem}, and can therefore be solved using either Algorithm~\ref{alg:better} or Algorithm~\ref{alg:plane}, depending on the dimension of the CSI vector. Once an optimal configuration \(\mathbf{x}_{\text{com}}^* : S_1 \to \{+1, -1\}\) is identified, it remains fixed throughout the whole sensing procedure.

\subsection{Sensing Setup}\label{sec:sensing_setup}

After identifying the parameters \(x_n, n \in S_1\) of the RIS, it is now time to choose the phase shifts that will facilitate the sensing task. For this, we follow a technique similar to~\cite{BuzziGrossiLopsVenturino}. We assign to each possible target position \(\boldsymbol{\psi}_t \in A\) a real value \(\eta\) that quantifies the confidence that the target is located at \(\boldsymbol{\psi}_t\). Then, the estimated target position is the point that maximizes this function \(\eta\):
\begin{equation}
\boldsymbol{\psi}_{\text{estimated}} \in \underset{\boldsymbol{\psi}_t \in A}{\mathbf{\argmax}} \ \eta .
\end{equation}

As will become evident, the function \(\eta\) is computationally expensive to evaluate, and its gradients are unknown. Therefore, we rely on gradient-free optimization methods to solve the above problem, such as Bayesian optimization. These methods are combined with an initial coarse grid search to guide the optimization process.

The confidence \(\eta\) is computed by comparing the measured CSI values with the expected CSI values assuming the target were located at position \(\boldsymbol{\psi}_t\). To this end, for each position \(\boldsymbol{\psi}_t\), we define the vectors:
\begin{align}
    \mathbf{h}^\text{T}_{t}  &= \left(
    \mathcal{C}\!\left(\boldsymbol{\psi}^\text{T}_1,\boldsymbol{\psi}_{t}\right),
    \mathcal{C}\!\left(\boldsymbol{\psi}^\text{T}_2,\boldsymbol{\psi}_{t}\right),
    \dots,
    \mathcal{C}\!\left(\boldsymbol{\psi}^\text{T}_{N_\text{T}},\boldsymbol{\psi}_{t}\right)
    \right)^\top , \\
    \mathbf{h}^\text{R}_{t}  &= \left(
    \mathcal{C}\!\left(\boldsymbol{\psi}_{t},\boldsymbol{\psi}^\text{R}_1\right),
    \mathcal{C}\!\left(\boldsymbol{\psi}_{t},\boldsymbol{\psi}^\text{R}_2\right),
    \dots,
    \mathcal{C}\!\left(\boldsymbol{\psi}_{t},\boldsymbol{\psi}^\text{R}_{N_\text{R}}\right)
    \right)^\top ,
\end{align}
with \(\mathbf{h}^\text{T}_{t} \in \mathbb{C}^{N_\text{T}}\) and \(\mathbf{h}^\text{R}_{t}  \in \mathbb{C}^{N_\text{R}}\), which correspond to the channels between the possible target position and the transmitter and receiver arrays. The expected CSI is then given by:
\begin{equation}
\mathbf{r}_\text{exp} =  \mathbf{r}_{\text{com}} + \alpha\left( \mathbf{h}^\text{T}_{t} \otimes \mathbf{h}^\text{R}_{t} + \sum_{n=1}^N x_n \mathcal{C}\!\left(\boldsymbol{\psi}^\text{RIS}_n,\boldsymbol{\psi}_t \right) \mathbf{h}^\text{T}_{n} \otimes \mathbf{h}^\text{R}_{t} \right).
\end{equation}

Since we know all the parameters that define \(\mathbf{r}_\text{com}\), we can subtract it from \(\mathbf{r}_\text{total}\), and obtain \(\mathbf{r}_\text{sen} + \mathbf{w}\). We then compare \(\mathbf{r}_{\text{sen}} + \mathbf{w}\) with the expected signal that contains only reflection paths including the target:
\begin{equation}
\mathbf{r}_{\text{exp},\text{sen}} =  \mathbf{h}^\text{T}_{t} \otimes \mathbf{h}^\text{R}_{t} + \sum_{n=1}^N x_n \mathcal{C}\!\left(\boldsymbol{\psi}^\text{RIS}_n,\boldsymbol{\psi}_t \right) \mathbf{h}^\text{T}_{n} \otimes \mathbf{h}^\text{R}_{t}.
\end{equation}
The comparison is done by taking the real part of the dot product of \(\mathbf{H} = \mathbf{r}_{\text{exp},\text{sen}}/\|\mathbf{r}_{\text{exp},\text{sen}}\|\) and \(\mathbf{r}_\text{sen} + \mathbf{w}\), hence setting:
\begin{equation}
\eta = \Re\!\left(\mathbf{H}^\dagger (\mathbf{r}_\text{sen}+\mathbf{w})\right)
\end{equation}
as the confidence value for \(\boldsymbol{\psi}_t\), where \((\cdot)^\dagger\) denotes the Hermitian operator.

Therefore, we need to choose \(x_n\), \(n \in S_2\), for each \(\boldsymbol{\psi}_t\) separately so as to maximize the value that \(\eta\) would have if the target were truly at point \(\boldsymbol{\psi}_t\). Assuming the target is at \(\boldsymbol{\psi}_t\) and ignoring the AWGN we get:
\begin{equation}
\eta = \frac{\Re(\mathbf{r}_{\text{exp},\text{sen}}^\dagger \mathbf{r}_{\text{sen}})}{\|\mathbf{r}_{\text{exp},\text{sen}}\|} = \frac{\|\mathbf{r}_{\text{exp},\text{sen}}\|^2}{\|\mathbf{r}_{\text{exp},\text{sen}}\|} = \|\mathbf{r}_{\text{exp},\text{sen}}\| ,
\end{equation}
and we have to maximize over \(x_n\), \(n \in S_2\), the norm of \(\mathbf{r}_{\text{exp},\text{sen}}\). This corresponds to the optimization problem:
\begin{equation}
\mathbf{x}^*_\text{sen} \in \underset{\substack{x_n \in \{+1,-1\} \\ n \in S_2}}{\mathbf{\argmax}} \left\| \mathbf{h}^\text{T}_{t} \otimes \mathbf{h}^\text{R}_{t} + \sum_{n=1}^N x_n \mathcal{C}\!\left(\boldsymbol{\psi}^\text{RIS}_n,\boldsymbol{\psi}_t \right) \mathbf{h}^\text{T}_{n} \otimes \mathbf{h}^\text{R}_{t} \right\| .
\end{equation}

Since the values \(x_n\), \(n \in S_1\), have already been determined by the method described in Section~\ref{sec:comm_setup}, we rewrite the above problem as:
\begin{align}
\mathbf{x}^*_\text{sen} &\in \underset{\substack{x_n \in \{+1,-1\} \\ n \in S_2}}{\mathbf{\argmax}} \left\| \mathbf{c} + \sum_{n \in S_2} x_n \mathcal{C}\!\left(\boldsymbol{\psi}^\text{RIS}_n,\boldsymbol{\psi}_t \right) \mathbf{h}^\text{T}_{n} \otimes \mathbf{h}^\text{R}_{t} \right\| , \\
\mathbf{c} &= \mathbf{h}^\text{T}_{t} \otimes \mathbf{h}^\text{R}_{t} + \sum_{n \in S_1} x_n \mathcal{C}\!\left(\boldsymbol{\psi}^\text{RIS}_n,\boldsymbol{\psi}_t \right) \mathbf{h}^\text{T}_{n} \otimes \mathbf{h}^\text{R}_{t} ,
\end{align}
and solve it using again Algorithms~\ref{alg:better} and~\ref{alg:plane}.

\section{Simulation Results}\label{sec:results}

In this section, we evaluate the proposed framework from two complementary perspectives. First, we assess whether the geometric structure identified in Section~\ref{sec:optimization} enables more effective exploration of the binary configuration space under a fixed sampling budget. Second, we evaluate the applicability of the same optimization framework to the communication and sensing functionalities introduced in Section~\ref{sec:isac}. All simulations are conducted in a three-dimensional environment, where the system geometry, RIS configuration, and channel model follow the description in Section~\ref{sec:system_model}. The optimization procedures introduced in Section~\ref{sec:optimization} are used to determine the RIS phase configurations under the considered scenarios.

We consider two simulation settings. The first focuses exclusively on LOS maximization and evaluates the geometry-informed sampling method against uninformed binary sampling and an ideal continuous-phase benchmark. The second considers the full ISAC scenario of Section~\ref{sec:isac}, where the RIS elements are partitioned between communication and sensing to assess the corresponding performance trade-off. For all simulations, the RIS is placed at position \(\boldsymbol{\psi}^{\mathrm{RIS}} = (0,0,5\lambda_c)\), with surface normal \(\widehat{\mathbf{n}} = (5,0,-1)/\sqrt{26}\) and tangent vector \(\widehat{\mathbf{t}} = (0,1,0)\). The spacing between adjacent RIS elements is set to \(d_\text{cell} = \lambda_c/5\), and the transmit power is fixed at \(P = 1~\mathrm{W}\). Whenever a sampling-based algorithm is employed, the sampling budget is fixed to \(N_{\text{samples}} = 1000\).

\subsection{LOS Maximization Evaluation}\label{sec:results_LOS}

We first evaluate the communication setup described in Section~\ref{sec:optimization}. We consider an \(8 \times 8\) RIS. The system consists of a linear array of \(N_\text{T}\) transmitters centered at \((0,0,2.5\lambda_c)\) and oriented parallel to the \(y\)-axis, and a linear array of \(N_\text{R}\) receivers centered at \((5\lambda_c,0,2.5\lambda_c)\), also oriented parallel to the \(y\)-axis. Adjacent elements in both arrays are spaced by \(\lambda_c/4\).

We compare the informed sampling approach of Algorithm~\ref{alg:better} with two benchmarks. The first is direct sampling, in which candidate configurations are drawn uniformly from the binary space \(\{+1,-1\}^N\). The second is an ideal continuous-phase RIS optimized using the method of~\cite{BuzziGrossiLopsVenturino}. The informed and direct binary methods use the same sampling budget, allowing the effect of restricting the search to geometrically admissible configurations to be isolated directly. The performance metric is the received power \(\|\mathbf{r}_\text{total}\|^2\). The number of transmitters \(N_\text{T}\) varies in \(\{1,5,9\}\), while the number of receivers \(N_\text{R}\) takes values in \(\{1,2,\dots,10\}\). The results are shown in Fig.~\ref{fig:LOS_evaluation1}.

\begin{figure*}[t]
\centering

\begin{tikzpicture}
\begin{axis}[
    height=0.55\columnwidth,
    width=1.25\columnwidth,
    scale only axis,
    xlabel={\(N_\text{R}\)},
    ylabel={\(\|\mathbf{r}_{\text{total}}\|^2\) \((\mathrm{dBW})\)},
    grid=both,
    legend style={
        at={(1.025,0.5)},
        anchor=west
    },
]


\addplot[solid, very thick, myblue, mark=o, mark options={solid}]
table [x expr=\coordindex+1, y expr=20*log10(\thisrowno{3}), col sep=comma]
{CSV_files/LOS_data.csv};
\addlegendentry{Informed sampling, \(N_\text{T} = 1\)}

\addplot[solid, very thick, mygreen, mark=square, mark options={solid}]
table [x expr=\coordindex+1, y expr=20*log10(\thisrowno{4}), col sep=comma]
{CSV_files/LOS_data.csv};
\addlegendentry{Informed sampling, \(N_\text{T} = 5\)}

\addplot[solid, very thick, myorange, mark=triangle, mark options={solid}]
table [x expr=\coordindex+1, y expr=20*log10(\thisrowno{5}), col sep=comma]
{CSV_files/LOS_data.csv};
\addlegendentry{Informed sampling, \(N_\text{T} = 9\)}


\addplot[dotted, very thick, myblue, mark=o, mark options={solid}]
table [x expr=\coordindex+1, y expr=20*log10(\thisrowno{0}), col sep=comma]
{CSV_files/LOS_data.csv};
\addlegendentry{Direct sampling, \(N_\text{T} = 1\)}

\addplot[dotted, very thick, mygreen, mark=square, mark options={solid}]
table [x expr=\coordindex+1, y expr=20*log10(\thisrowno{1}), col sep=comma]
{CSV_files/LOS_data.csv};
\addlegendentry{Direct sampling, \(N_\text{T} = 5\)}

\addplot[dotted, very thick, myorange, mark=triangle, mark options={solid}]
table [x expr=\coordindex+1, y expr=20*log10(\thisrowno{2}), col sep=comma]
{CSV_files/LOS_data.csv};
\addlegendentry{Direct sampling, \(N_\text{T} = 9\)}


\addplot[dashed, very thick, myblue, mark=o, mark options={solid}]
table [x expr=\coordindex+1, y expr=20*log10(\thisrowno{6}), col sep=comma]
{CSV_files/LOS_data.csv};
\addlegendentry{Continuous RIS, \(N_\text{T} = 1\)}

\addplot[dashed, very thick, mygreen, mark=square, mark options={solid}]
table [x expr=\coordindex+1, y expr=20*log10(\thisrowno{7}), col sep=comma]
{CSV_files/LOS_data.csv};
\addlegendentry{Continuous RIS, \(N_\text{T} = 5\)}

\addplot[dashed, very thick, myorange, mark=triangle, mark options={solid}]
table [x expr=\coordindex+1, y expr=20*log10(\thisrowno{8}), col sep=comma]
{CSV_files/LOS_data.csv};
\addlegendentry{Continuous RIS, \(N_\text{T} = 9\)}

\end{axis}
\end{tikzpicture}

\caption{Received total power, \(\|\mathbf{r}_{\text{total}}\|^2\), by \(N_\text{R}\) receiving antennas for \(1\) (blue color), \(5\) (green color), and \(9\) (orange color) transmitting antennas, \(N_\text{T}\), as evaluated by Algorithm~\ref{alg:better} for the binary RIS (solid line), and compared with direct sampling for the binary RIS (dotted line) and the continuous-phase RIS benchmark (dashed line).}

\label{fig:LOS_evaluation1}
\end{figure*}
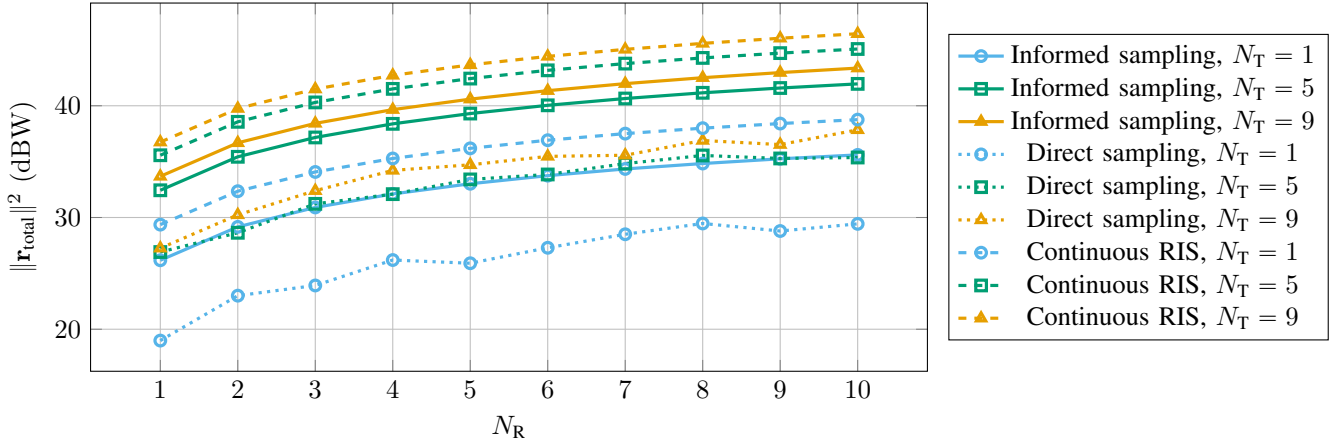

Figure~\ref{fig:LOS_evaluation1} shows that restricting the search to configurations satisfying the geometric condition of Lemma~\ref{lem:essential_x} leads to a substantial improvement over uninformed sampling. Under the same budget of \(N_{\text{samples}}=1000\) candidate evaluations, the proposed method achieves gains of at least \(5.36~\mathrm{dB}\) over direct sampling across the considered configurations, reaching \(7.11~\mathrm{dB}\) for \(N_\text{T}=1\) and \(N_\text{R}=5\). Therefore, the improvement does not result from evaluating more configurations, but from using the available sampling budget exclusively on structurally admissible candidates.

The informed-sampling curves also exhibit considerably less fluctuation than those obtained through direct sampling as \(N_\text{R}\) varies. This behavior indicates that the geometric restriction consistently generates higher-quality candidates under the same computational budget, whereas uninformed sampling remains more sensitive to the random exploration of the exponentially large binary space.

As expected, the ideal continuous-phase RIS achieves the highest received power. Nevertheless, the geometry-informed 1-bit RIS remains within \(3.12~\mathrm{dB}\) of the continuous-phase benchmark. This result indicates that effective exploration of the binary configuration space can recover a substantial fraction of the performance offered by ideal continuous phase control, despite the severe 1-bit hardware constraint.

\subsection{ISAC Evaluation}

We next evaluate the proposed framework in the ISAC setting described in Section~\ref{sec:isac}. The objective of this experiment is not only to quantify the communication--sensing trade-off, but also to demonstrate that the same binary optimization principle developed in Section~\ref{sec:optimization} can be applied to both functionalities.

The sensing area \(A\) is defined as a cube with side length \(5\lambda_c\):
\begin{equation}
    A = [0,5\lambda_c] \times [-2.5\lambda_c,2.5\lambda_c] \times [0,5\lambda_c].
\end{equation}
We consider \(N_\text{T}=1\) transmitting antenna positioned at \(\boldsymbol{\psi}^\text{T}_1=(0,0,2.5\lambda_c)\), and \(N_\text{R}=3\) receiving antennas positioned at \(\boldsymbol{\psi}^\text{R}_1=(5\lambda_c,1.25\lambda_c,0)\), \(\boldsymbol{\psi}^\text{R}_2=(2.5\lambda_c,2.5\lambda_c,5\lambda_c)\), and \(\boldsymbol{\psi}^\text{R}_3=(5\lambda_c,-2.5\lambda_c,2.5\lambda_c)\). The reflection coefficient is set to \(\alpha=0.8\).

We consider two RIS grid sizes, \(8 \times 8\) and \(16 \times 16\). For each RIS size, nine allocations of \((S_1,S_2)\) are evaluated, ranging uniformly from \((0,N)\), corresponding to full sensing, to \((N,0)\), corresponding to full communication, where the total number of RIS elements \(N\) is either \(64\) or \(256\). A set of \(100\) random target positions is generated within the sensing region \(A\), and the same target set is used across all RIS allocations to ensure comparability. Target position estimation is performed by maximizing the confidence function \(\eta\) through a hybrid optimization procedure. Specifically, \(\eta\) is first evaluated over a coarse grid to identify promising regions of the search space, after which Bayesian optimization is applied locally to refine the estimate in continuous space.

Communication performance is evaluated through the average received power. We consider both the power of the complete received signal, \(\mathbf{r}_\text{total}\) and the power of the signal without the reflections from the target \(\mathbf{r}_\text{com}\), as defined in Eq.~\eqref{eq:r}.
The corresponding metrics, \(\|\mathbf{r}_{\text{total}}\|^2\) and \(\|\mathbf{r}_{\text{com}}\|^2\), are averaged over all evaluations of \(\eta\) and reported in \(\mathrm{dBW}\).  Sensing performance is evaluated through the average localization error and the percentage of missed detections, where a detection is considered missed when its localization error exceeds the threshold \(T_e=\lambda_c/4\). The results are presented in Fig.~\ref{fig:ISAC_evaluation}.

\begin{figure}[ht]

\centering

\begin{tikzpicture}

\begin{groupplot}[
    group style={
        group size=1 by 3,
        vertical sep=0.1cm,
    },
    width=0.95\columnwidth,
    height=0.5\columnwidth,
    grid=both,
    yticklabel style={
        text width=1.8em,
        align=right
    },
    ylabel style={
        xshift=0em
    },
    legend columns=1,
    legend to name={mylegend},
    legend style={
        at={(2.025,0.5)},
        anchor=south
    },
]

\nextgroupplot[
    ylabel={\small{\(\mathrm{Error}~(\lambda_c)\)}},
    xtick={1,2,3,4,5,6,7,8,9},
    xticklabels={},
    ymode=log,
]

\addplot[
    solid, very thick, myorange, mark=triangle, mark options={solid}
]
table [x expr=\coordindex+1, y expr=2.5*\thisrowno{0}, col sep=comma]
{CSV_files/errors_8.csv};

\addplot[
    solid, very thick, myblue, mark=square, mark options={solid}
]
table [x expr=\coordindex+1, y expr=2.5*\thisrowno{0}, col sep=comma]
{CSV_files/errors_16.csv};

\addplot[
    dotted,
    black,
    very thick,
    domain=1:9,
] {0.25};

\node[
    draw,
    fill=white,
    font=\scriptsize,
    anchor=north west,
    align=left
] at (rel axis cs:0.02,0.98)
{
\begin{tabular}{@{}l@{}}
\tikz[baseline=4.1ex]{\draw[dotted, very thick] (0,0)--(0.55,0);}
\(T_e\)
\end{tabular}
};

\nextgroupplot[
    ylabel={\small{\(\mathrm{Power~(dBW)}\)}},
    xtick={1,2,3,4,5,6,7,8,9},
    xticklabels={},
]

\addplot[
    solid, very thick, myorange, mark=triangle, mark options={solid}
]
table [x expr=\coordindex+1, y expr=20*log10(\thisrowno{0}), col sep=comma]
{CSV_files/powers_8.csv};

\addplot[
    dotted, very thick, myorange, mark=triangle, mark options={solid}
]
table [x expr=\coordindex+1, y expr=20*log10(\thisrowno{0}), col sep=comma]
{CSV_files/powers_com_8.csv};

\addplot[
    solid, very thick, myblue, mark=square, mark options={solid}
]
table [x expr=\coordindex+1, y expr=20*log10(\thisrowno{0}), col sep=comma]
{CSV_files/powers_16.csv};

\addplot[
    dotted, very thick, myblue, mark=square, mark options={solid}
]
table [x expr=\coordindex+1, y expr=20*log10(\thisrowno{0}), col sep=comma]
{CSV_files/powers_com_16.csv};

\node[
    draw,
    fill=white,
    font=\scriptsize,
    anchor=north west,
    align=left
] at (rel axis cs:0.02,0.98)
{
\begin{tabular}{@{}l@{}}
\tikz{\draw[solid, very thick] (0,0)--(0.55,0);}
\(\|\mathbf{r}_{\text{total}}\|^2\)\\
\tikz{\draw[dotted, very thick] (0,0)--(0.55,0);}
\(\|\mathbf{r}_{\text{com}}\|^2\)
\end{tabular}
};

\nextgroupplot[
    ylabel={\small{\(\mathrm{Percentage~of~misses}\)}},
    xlabel={\small{\(\mathrm{Percentage~of~RIS~elements~used~for~communication}\)}},
    xtick={1,2,3,4,5,6,7,8,9},
    xticklabels={
        \(0\),\(\),\(25\),\(\),\(50\),
        \(\),\(75\),\(\),\(100\)
    },
    ytick={0,25,50,75,100},
]

\addplot[
    forget plot, solid, very thick, myorange, mark=triangle, mark options={solid}
]
table [x expr=\coordindex+1, y expr=\thisrowno{0}*100, col sep=comma]
{CSV_files/percentages_8.csv};

\addplot[
    forget plot, solid, very thick, myblue, mark=square, mark options={solid}
]
table [x expr=\coordindex+1, y expr=\thisrowno{0}*100, col sep=comma]
{CSV_files/percentages_16.csv};

\addlegendimage{solid, very thick, myblue, mark=square, mark options={solid}}
\addlegendentry{RIS grid size: \(16 \times 16\)}

\addlegendimage{solid, very thick, myorange, mark=triangle, mark options={solid}}
\addlegendentry{RIS grid size: \(8 \times 8\)}

\end{groupplot}

\node[
    draw,
    fill=white,
    font=\scriptsize,
    align=left,
    anchor=south
] at ([yshift=2mm]group c1r1.north)
{
\begin{tabular}{@{}l@{}}
RIS grid size:\quad
\tikz[baseline=-0.6ex]{
    \draw[very thick,myblue] (0,0)--(0.5,0);
}
\(16 \times 16\)
\quad
\tikz[baseline=-0.6ex]{
    \draw[very thick,myorange] (0,0)--(0.5,0);
}
\(8 \times 8\)
\end{tabular}
};

\end{tikzpicture}

\caption{Localization error (top plot, with the dotted line indicating the error threshold \(T_e\)), power of the total received vector, \(\|\mathbf{r}_\text{total}\|^2\) (middle plot, solid lines), power considering only reflections from the communication-oriented RIS elements, \(\|\mathbf{r}_\text{com}\|^2\) (middle plot, dotted lines), and percentage of missed detections with localization error larger than \(T_e\) (bottom plot), as functions of the percentage of RIS elements dedicated to communication. Each metric is shown for RIS grid sizes of \(8 \times 8\) (orange color) and \(16 \times 16\) (blue color).}

\label{fig:ISAC_evaluation}
\end{figure}
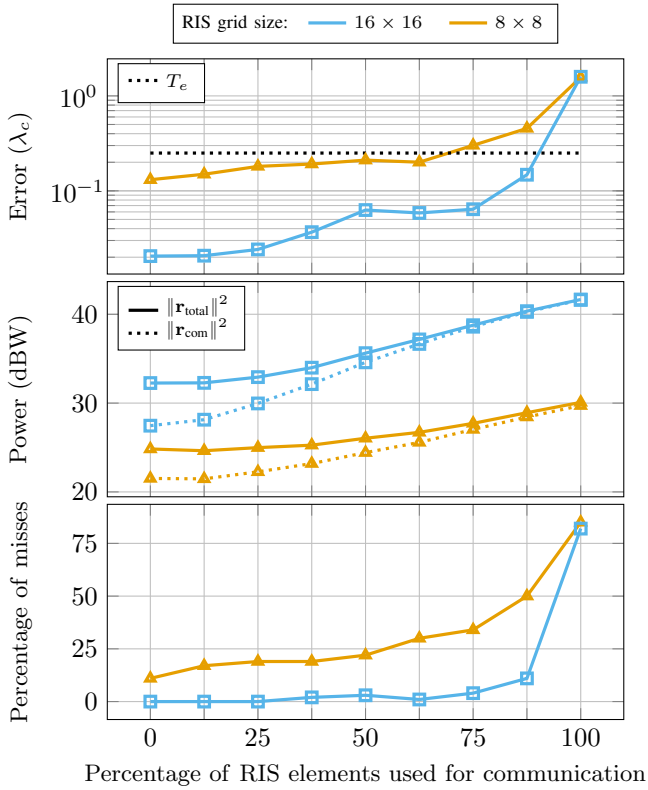

Figure~\ref{fig:ISAC_evaluation} illustrates the trade-off induced by allocating RIS elements between the two functionalities. As the fraction of communication-oriented elements increases, both \(\|\mathbf{r}_{\text{total}}\|^2\) and \(\|\mathbf{r}_{\text{com}}\|^2\) increase, while localization performance gradually degrades. The received power increases from \(24.83~\mathrm{dBW}\) and \(32.26~\mathrm{dBW}\) to \(30.08~\mathrm{dBW}\) and \(41.66~\mathrm{dBW}\) for the \(8 \times 8\) and \(16 \times 16\) RIS configurations, respectively. Furthermore, the contribution of the communication-oriented elements progressively approaches the total received power and becomes almost identical to it when all RIS elements are assigned to communication. The opposite trend is observed for sensing. In the full-sensing configuration, the miss ratios are \(11\%\) and \(0\%\) for the \(8 \times 8\) and \(16 \times 16\) RISs, respectively, increasing to \(85\%\) and \(82\%\) in the full-communication configuration. Similarly, the average localization error increases from \(0.13\lambda_c\) and \(0.02\lambda_c\) to \(1.57\lambda_c\) and \(1.59\lambda_c\), respectively.

More importantly, the results show that the communication--sensing trade-off is not necessarily proportional to the fraction of RIS elements assigned to each task. In particular, the miss-percentage curves remain relatively flat for low and moderate communication allocations before increasing rapidly as the system approaches the full communication mode. In contrast, the received power increases steadily across the entire allocation range, following an approximately linear trend (when expressed in dBW). This behavior reveals an operating region in which communication performance can be improved without a proportional degradation of sensing performance. In particular, for  the \(8 \times 8\) RIS with \(32\) elements allocated to communication and \(32\) elements allocated to sensing, we obtain \(\|\mathbf{r}_{\text{total}}\|^2 = 26.04~\mathrm{dBW}\) and \(\|\mathbf{r}_{\text{com}}\|^2 = 24.43~\mathrm{dBW}\), while maintaining an average localization error of \(0.21\lambda_c\) and a miss ratio of \(22\%\). Similarly, for the \(16 \times 16\) RIS with \(160\) communication and \(96\) sensing elements, we obtain \(\|\mathbf{r}_{\text{total}}\|^2 = 37.17~\mathrm{dBW}\), \(\|\mathbf{r}_{\text{com}}\|^2 = 36.66~\mathrm{dBW}\), an average localization error of \(0.058\lambda_c\), and a miss ratio of only \(1\%\).

Overall, the ISAC evaluation demonstrates that the geometric binary optimization framework is not confined to communication-oriented LOS enhancement. Although communication and sensing involve different channel-dependent vectors and operational objectives, both can be handled through the same optimization principle, while the RIS partition provides a direct mechanism for controlling their performance trade-off.

\section{Conclusions}\label{sec:conclusions}

This work investigated 1-bit RIS optimization from a geometric perspective. Rather than treating the \(2^N\) binary configurations as an unstructured search space, we showed that every globally optimal solution must satisfy a common geometric condition, with its binary signs induced by the projections of channel-dependent vectors onto a common direction. For general MIMO systems, this structure was exploited through a geometry-informed sampling method that restricts the search to structurally admissible configurations, while in the SISO case the same principle reduces to a two-dimensional angular partition that allows the complete candidate set to be characterized and the global optimum to be recovered through polynomial-time enumeration. The numerical results confirm the benefit of exploiting this structure over uninformed binary sampling while maintaining performance close to the continuous-phase benchmark. Finally, the same optimization principle was applied to an ISAC scenario, where communication enhancement and target localization reduce to the same underlying binary geometric problem, providing a basis for exploring RIS-assisted sensing in more complex indoor environments. Overall, the results show that the discrete nature of practical RIS hardware does not render the phase optimization intractable, but it reveals exploitable structure that can support efficient implementation across different RIS functionalities.

\bibliographystyle{IEEEtran}
\bibliography{my_bib}

@misc{LiaskosEtAl,
      title={A Tutorial on Controlling Metasurfaces from the Network Perspective}, 
      author={Christos Liaskos and others},
      year={2026},
      eprint={2601.12118},
      archivePrefix={arXiv},
      primaryClass={cs.NI},
      url={https://arxiv.org/abs/2601.12118}, 
}

@ARTICLE{LiaskosEtAl2,
  author={Liaskos, Christos and others},
  journal={Proceedings of the IEEE}, 
  title={Software-Defined Reconfigurable Intelligent Surfaces: From Theory to End-to-End Implementation}, 
  year={2022},
  volume={110},
  number={9},
  pages={1466-1493},
  doi={10.1109/JPROC.2022.3169917}}

@misc{PutrantoHamzaMabroukiArmiDayoup,
      title={Reconfigurable Intelligent Surfaces for {6G} and Beyond: A Comprehensive Survey from Theory to Deployment}, 
      author={Prasetyo Putranto and Anis Amazigh Hamza and Sameh Mabrouki and Nasrullah Armi and Iyad Dayoub},
      year={2025},
      eprint={2506.19526},
      archivePrefix={arXiv},
      primaryClass={eess.SP},
      url={https://arxiv.org/abs/2506.19526}, 
}

@ARTICLE{PapadopoulosEtAl,
  author={Papadopoulos, Alexandros I. and others},
  journal={IEEE Communications Magazine}, 
  title={Physics-Informed Metaheuristics for Fast {RIS} Codebook Compilation}, 
  year={2024},
  volume={62},
  number={11},
  pages={152-158},
  doi={10.1109/MCOM.001.2300582}}

@ARTICLE{PapadopoulosEtAl2,
  author={Papadopoulos, Alexandros I. and others},
  journal={IEEE Open Journal of the Communications Society}, 
  title={{RF}-Fencing: A Novel {RIS}-Based Service for Signal Suppression via Codebook Multiplexing}, 
  year={2026},
  volume={},
  number={},
  pages={1-1},
  doi={10.1109/OJCOMS.2026.3710143}}

@ARTICLE{PapadopoulosEtAl3,
  author={Papadopoulos, Alexandros and others},
  journal={IEEE Transactions on Network and Service Management}, 
  title={On Modeling the {RIS} as a Resource: Multi-User Allocation and Efficiency-Proportional Pricing}, 
  year={2025},
  volume={22},
  number={5},
  pages={4694-4705},
  doi={10.1109/TNSM.2025.3576038}}

@ARTICLE{WuZhang,
  author={Wu, Qingqing and Zhang, Rui},
  journal={IEEE Communications Magazine}, 
  title={Towards Smart and Reconfigurable Environment: Intelligent Reflecting Surface Aided Wireless Network}, 
  year={2020},
  volume={58},
  number={1},
  pages={106-112},
  doi={10.1109/MCOM.001.1900107}}

@article{BuzziGrossiLopsVenturino,
  title={Foundations of {MIMO} radar detection aided by reconfigurable intelligent surfaces},
  author={Buzzi, Stefano and Grossi, Emanuele and Lops, Marco and Venturino, Luca},
  journal={IEEE Transactions on Signal Processing},
  volume={70},
  pages={1749--1763},
  year={2022},
  month = {March},
  publisher={IEEE},
  doi = {10.1109/TSP.2022.3157975}
}

@ARTICLE{TishchenkoEtAl,
  author={Tishchenko, Anton and others},
  journal={IEEE Communications Surveys and Tutorials}, 
  title={The Emergence of Multi-Functional and Hybrid Reconfigurable Intelligent Surfaces for Integrated Sensing and Communications - A Survey}, 
  year={2025},
  volume={27},
  number={5},
  pages={2895-2936}}

@ARTICLE{ShekhawatKashyapRaldirisShanTrichopoulos,
  author={Shekhawat, Aditya S. and Kashyap, Bharath G. and Raldiris Torres, Russell W. and Shan, Feiyu and Trichopoulos, Georgios C.},
  journal={IEEE Open Journal of Antennas and Propagation}, 
  title={A Millimeter-Wave Single-Bit Reconfigurable Intelligent Surface With High-Resolution Beam-Steering and Suppressed Quantization Lobe}, 
  year={2025},
  volume={6},
  number={1},
  pages={311-325}}

@ARTICLE{ZhaoJianChenZhaoMu,
  author={Zhao, Xianming and Jian, Mengnan and Chen, Yijian and Zhao, Yajun and Mu, Lin},
  journal={Intelligent and Converged Networks}, 
  title={Reconfigurable Intelligent Surfaces for {6G}: Engineering Challenges and the Road Ahead}, 
  year={2025},
  volume={6},
  number={1},
  pages={53-81}}

@ARTICLE{JohariEtAl,
  author={Johari, Safpbri and others},
  journal={IEEE Access}, 
  title={Design and {SDR} Validation of a 5.8 GHz 1-bit Reconfigurable Intelligent Surface With Optimized {RF} Choke}, 
  year={2025},
  volume={13},
  number={},
  pages={182556-182568},
  doi={10.1109/ACCESS.2025.3624689}}

@ARTICLE{WuZhang2,
  author={Wu, Qingqing and Zhang, Rui},
  journal={IEEE Transactions on Wireless Communications}, 
  title={Intelligent Reflecting Surface Enhanced Wireless Network via Joint Active and Passive Beamforming}, 
  year={2019},
  volume={18},
  number={11},
  pages={5394-5409},
  doi={10.1109/TWC.2019.2936025}}

@ARTICLE{HuangZapponeAlexandropoulosDebbahYuen,
  author={Huang, Chongwen and Zappone, Alessio and Alexandropoulos, George C. and Debbah, Mérouane and Yuen, Chau},
  journal={IEEE Transactions on Wireless Communications}, 
  title={Reconfigurable Intelligent Surfaces for Energy Efficiency in Wireless Communication}, 
  year={2019},
  volume={18},
  number={8},
  pages={4157-4170},
  doi={10.1109/TWC.2019.2922609}}

@INPROCEEDINGS{YuXuSchober,
  author={Yu, Xianghao and Xu, Dongfang and Schober, Robert},
  booktitle={2019 IEEE/CIC International Conference on Communications in China (ICCC)}, 
  title={{MISO} Wireless Communication Systems via Intelligent Reflecting Surfaces : (Invited Paper)}, 
  year={2019},
  volume={},
  number={},
  pages={735-740},
  doi={10.1109/ICCChina.2019.8855810}}

@ARTICLE{PerovicTranRenzoFlanagan,
  author={Perović, Nemanja Stefan and Tran, Le-Nam and Di Renzo, Marco and Flanagan, Mark F.},
  journal={IEEE Transactions on Wireless Communications}, 
  title={Achievable Rate Optimization for {MIMO} Systems With Reconfigurable Intelligent Surfaces}, 
  year={2021},
  volume={20},
  number={6},
  pages={3865-3882},
  doi={10.1109/TWC.2021.3054121}}

@INPROCEEDINGS{HuangAlexandropoulosZapponeDebbahYuen,
  author={Huang, Chongwen and Alexandropoulos, George C. and Zappone, Alessio and Debbah, Mérouane and Yuen, Chau},
  booktitle={2018 IEEE Globecom Workshops (GC Wkshps)}, 
  title={Energy Efficient Multi-User {MISO} Communication Using Low Resolution Large Intelligent Surfaces}, 
  year={2018},
  volume={},
  number={},
  pages={1-6},
  doi={10.1109/GLOCOMW.2018.8644519}}

@ARTICLE{ZhaoWuZhaoZhang,
  author={Zhao, Ming-Min and Wu, Qingqing and Zhao, Min-Jian and Zhang, Rui},
  journal={IEEE Transactions on Wireless Communications}, 
  title={Intelligent Reflecting Surface Enhanced Wireless Networks: Two-Timescale Beamforming Optimization}, 
  year={2021},
  volume={20},
  number={1},
  pages={2-17},
  doi={10.1109/TWC.2020.3022297}}

@ARTICLE{YuanLiangJoungFengLarson,
  author={Yuan, Jie and Liang, Ying-Chang and Joung, Jingon and Feng, Gang and Larsson, Erik G.},
  journal={IEEE Transactions on Communications}, 
  title={Intelligent Reflecting Surface-Assisted Cognitive Radio System}, 
  year={2021},
  volume={69},
  number={1},
  pages={675-687},
  doi={10.1109/TCOMM.2020.3033006}}

@ARTICLE{NiLiuYangTianShen,
  author={Ni, Wanli and Liu, Yuanwei and Yang, Zhaohui and Tian, Hui and Shen, Xuemin},
  journal={IEEE Transactions on Wireless Communications}, 
  title={Integrating Over-the-Air Federated Learning and Non-Orthogonal Multiple Access: What Role Can {RIS} Play?}, 
  year={2022},
  volume={21},
  number={12},
  pages={10083-10099},
  doi={10.1109/TWC.2022.3181214}}

@ARTICLE{ZengZhangDiHanSong,
  author={Zeng, Shuhao and Zhang, Hongliang and Di, Boya and Han, Zhu and Song, Lingyang},
  journal={IEEE Communications Letters}, 
  title={Reconfigurable Intelligent Surface ({RIS}) Assisted Wireless Coverage Extension: {RIS} Orientation and Location Optimization}, 
  year={2021},
  volume={25},
  number={1},
  pages={269-273},
  doi={10.1109/LCOMM.2020.3025345}}

@INPROCEEDINGS{WuZhang3,
  author={Wu, Qingqing and Zhang, Rui},
  booktitle={ICASSP 2019 - 2019 IEEE International Conference on Acoustics, Speech and Signal Processing (ICASSP)}, 
  title={Beamforming Optimization for Intelligent Reflecting Surface with Discrete Phase Shifts}, 
  year={2019},
  volume={},
  number={},
  pages={7830-7833},
  doi={10.1109/ICASSP.2019.8683145}}

@INPROCEEDINGS{YouZhengZhang,
  author={You, Changsheng and Zheng, Beixiong and Zhang, Rui},
  booktitle={ICC 2020 - 2020 IEEE International Conference on Communications (ICC)}, 
  title={Intelligent Reflecting Surface with Discrete Phase Shifts: Channel Estimation and Passive Beamforming}, 
  year={2020},
  volume={},
  number={},
  pages={1-6},
  doi={10.1109/ICC40277.2020.9149292}}

@INPROCEEDINGS{StylianopoulosGavriilidisAlexandropoulos,
  author={Stylianopoulos, Kyriakos and Gavriilidis, Panagiotis and Alexandropoulos, George C.},
  booktitle={2024 IEEE 25th International Workshop on Signal Processing Advances in Wireless Communications (SPAWC)}, 
  title={Asymptotically Optimal Closed-Form Phase Configuration of 1-bit {RISs} via Sign Alignment}, 
  year={2024},
  volume={},
  number={},
  pages={746-750},
  doi={10.1109/SPAWC60668.2024.10694052}}

@ARTICLE{VejlingKimBiscioWymeerschPopovski,
  author={Vejling, Martin V. and Kim, Hyowon and Biscio, Christophe A. N. and Wymeersch, Henk and Popovski, Petar},
  journal={IEEE Transactions on Signal Processing}, 
  title={{RIS}-Assisted High Resolution Radar Sensing}, 
  year={2025},
  volume={73},
  number={},
  pages={2940-2955},
  doi={10.1109/TSP.2025.3586551}}

@ARTICLE{ShaoYouMaChenZhang,
  author={Shao, Xiaodan and You, Changsheng and Ma, Wenyan and Chen, Xiaoming and Zhang, Rui},
  journal={IEEE Journal on Selected Areas in Communications}, 
  title={Target Sensing With Intelligent Reflecting Surface: Architecture and Performance}, 
  year={2022},
  volume={40},
  number={7},
  pages={2070-2084},
  doi={10.1109/JSAC.2022.3155546}}

@misc{SotiropoulosEtAl,
      title={{RIS}-Assisted 3D Spherical Splatting for Object Composition Visualization using Detection Transformers}, 
      author={Anastasios T. Sotiropoulos and others},
      year={2025},
      eprint={2511.02573},
      archivePrefix={arXiv},
      primaryClass={eess.SP},
      url={https://arxiv.org/abs/2511.02573}, 
}

@ARTICLE{HuangYangTangWenXiaJin,
  author={Huang, Yixuan and Yang, Jie and Tang, Wankai and Wen, Chao-Kai and Xia, Shuqiang and Jin, Shi},
  journal={IEEE Transactions on Wireless Communications}, 
  title={Joint Localization and Environment Sensing by Harnessing {NLOS} Components in {RIS}-Aided mmWave Communication Systems}, 
  year={2023},
  volume={22},
  number={12},
  pages={8797-8813},
  doi={10.1109/TWC.2023.3266343}}

@ARTICLE{ElzanatyGuerraGuidiAlouni,
  author={Elzanaty, Ahmed and Guerra, Anna and Guidi, Francesco and Alouini, Mohamed-Slim},
  journal={IEEE Transactions on Signal Processing}, 
  title={Reconfigurable Intelligent Surfaces for Localization: Position and Orientation Error Bounds}, 
  year={2021},
  volume={69},
  number={},
  pages={5386-5402},
  doi={10.1109/TSP.2021.3101644}}

@ARTICLE{WuEtAl,
  author={Wu, Qingqing and others},
  journal={IEEE Communications Surveys and Tutorials}, 
  title={Intelligent Reflecting Surfaces for Integrated Sensing and Communications: From System Coexistence to Networked Mutualism}, 
  year={2026},
  volume={28},
  number={},
  pages={6057-6100},
  doi={10.1109/COMST.2026.3686910}}

@INPROCEEDINGS{ZhangRenShenLuo,
  author={Zhang, Yaowen and Ren, Shuyi and Shen, Kaiming and Luo, Zhi-Quan},
  booktitle={2021 IEEE Global Communications Conference (GLOBECOM)}, 
  title={Optimal Discrete Beamforming for Intelligent Reflecting Surface}, 
  year={2021},
  volume={},
  number={},
  pages={1-6},
  doi={10.1109/GLOBECOM46510.2021.9685790}}

@inbook{LiSunGuGaoLiu,
author = {Li, Duan and Sun, Xiaoling and Gu, Shenshen and Gao, Jianjun and Liu, Chunli},
year = {2010},
month = {05},
pages = {199-225},
title = {Polynomially Solvable Cases of Binary Quadratic Programs},
volume = {39},
isbn = {978-0-387-89495-9},
doi = {10.1007/978-0-387-89496-6_11}
}

@ARTICLE{MunteLazaroVillarino,
  author={Munte, Nil and Lazaro, Antonio and Villarino, Ramon and Girbau, David},
  journal={IEEE Sensors Journal}, 
  title={Vehicle Occupancy Detector Based on FMCW mm-Wave Radar at 77 GHz}, 
  year={2022},
  volume={22},
  number={24},
  pages={24504-24515},
  doi={10.1109/JSEN.2022.3218454}}

@ARTICLE{SatoWandaleIchigeKimuraSugiura,
  author={Sato, Kotone and Wandale, Steven and Ichige, Koichi and Kimura, Kazuya and Sugiura, Ryo},
  journal={IEEE Sensors Journal}, 
  title={Millimeter-Wave Radar-Based Vehicle In-Cabin Occupancy Detection Using Explainable Machine Learning}, 
  year={2024},
  volume={24},
  number={15},
  pages={24288-24298},
  doi={10.1109/JSEN.2024.3413775}}

\end{document}